\documentclass[11pt]{amsart}

\usepackage{amsmath}
\usepackage{amssymb}
\usepackage{amsfonts}
\usepackage{amsthm}
\usepackage{mathabx}
\usepackage{mathrsfs}
\usepackage{dirtytalk}

\newcommand*\norm[1]{ \left\| #1 \right\| }

\newcommand*\Z{ \mathbb{Z} }

\newcommand*\R{ \mathbb{R} }
\newcommand*\C{ \mathbb{C} }
\newcommand*\N{ \mathbb{N} }

\newcommand*\T{\mathbb T }
\newcommand\set[1]{\left\{ #1 \right\}}

\newcommand\EQ[1]{\begin{equation} #1 \end{equation}}

\theoremstyle{definition}

\newtheorem{Remark}{Remark}

\theoremstyle{theorem}

\newtheorem{mythm}{Theorem}[section]

\newtheorem{mylemma}{Lemma}[section]

\newtheorem{mycor}{Corollary}[section]

\title[Pointwise modulus of continuity of the Lyapunov exponent]{Pointwise modulus of continuity of the {L}yapunov exponent and integrated density of states for analytic multi-frequency quasiperiodic $M(2,\C)$ cocycles}
\author{Matthew Powell}
\address{Department of Mathematics, Georgia Institute of Technology, Atlanta GA, 30332}
\date{\today}

\begin{document}
\maketitle

\begin{abstract}
It is known that the Lyapunov exponent for multifrequency analytic cocycles is weak-H\"older continuous in cocycle for certain Diophantine frequencies, and that this implies certain regularity of the integrated density of states in energy for Jacobi operators. In this paper, we establish the pointwise modulus of continuity in both cocycle and frequency and obtain analogous regularity of the integrated density of states in energy, potential, and frequency.
\end{abstract}

\section{Introduction}

In this paper, we are interested in the regularity of the Lyapunov exponent associated to multifrequency quasi-periodic cocycles. Let $M(2,\C)$ denote the set of $2 \times 2$ matrices with complex entries. Let $\T^d = \R^d/ \Z^d$ denote the $d$-dimensional torus and let $S_\omega: \T^d \to \T^d$ denote the shift by $\omega \in \T^d:$ $x \mapsto x + \omega.$ A $d$-dimensional quasi-periodic cocycle is a pair $(A, \omega) \in C(\T^d, M(2,\C)) \times \R$ understood as a linear skew product
$(A,\omega): \C^2 \times \T^d \to \C^2\times \T^d$ with
$$(w,x) \mapsto (A(x)w, S_\omega x).$$
Cocycles enjoy the property that they may be iterated, in the following sense: the $N^{th}$ iterate of $(A,\omega)$ is
$$A_N(x,\omega) = \prod_{j = N-1}^0 A(S^j_\omega x).$$

We are interested in {\it analytic quasiperiodic cocycles}, so we assume $A$ is an analytic $M(2,\C)$-valued function on $\T^d.$ It is known that continuity of the Lyapunov exponent fails, in general, when $A$ is only $C^\infty$ (c.f. \cite{MR2825743} or \cite{WangYou}). Since we assume that $A$ is analytic on $\T^d,$ we may extend $A$ to some complex strip, $|\Im(z_j)| < \rho, j \leq d.$ We denote the space of such $A$ by $C_\rho(\T^d, M(2,\C)),$ and imbue this with the natural metric
$$\norm{A}_\rho := \sup_{|\Im(z_j)| < \rho/2} \norm{A(z_1,...,z_d)}.$$
\begin{Remark}
Note that we take the supremum over $|\Im(z_j)| < \rho/2$ rather than over $|\Im(z_j)| < \rho$ to ensure that the supremum exists. We may do this because $A$ is also an analytic cocycle over the closed set $\set{z\in \C^d: |\Im(z_j)| \leq \rho/2}.$
\end{Remark}
From this, we can inductively define a topology on the space of {\it all} analytic cocycles, but since convergence in this topology is equivalent to convergence in $(C_\rho, \norm{\cdot}_\rho)$ for some $\rho >0,$ we will restrict our attention to $C_\rho$ for a fixed $\rho.$

Cocycles of this form have been used extensively to study one-dimensional discrete Jacobi operators: $$H_{x,\omega} : \ell^2(\Z) \to \ell^2(\Z)$$ given by
\EQ{\label{eq:JacobiOp}
(H_{x,\omega}\psi)(n) = \overline{a(S^{n-1}_\omega x)}\psi(n - 1) + a_(S^{n}_\omega x) \psi(n + 1) + v(S^n_\omega x)\psi(n),}
where $a, v\in C(\T^d, \R).$ Solutions to the eigenequation $H_{x,\omega} \psi = E\psi$ may be recovered via the $N^{th}$ transfer matrix
$$A_N(x,\omega,E) = \prod_{j = N - 1}^0 M_j(x,\omega,E),$$
where
$$M_j(x,\omega,E) = \begin{pmatrix} E - v(S^{j+1}_\omega x) & -\overline{a(S^{j}_\omega x)} \\ a(S^{j+ 1}_\omega x) & 0\end{pmatrix}.$$
Indeed, any solution to $H_{x,\omega} \psi = E\psi$ satisfies
$$\begin{pmatrix} \psi(N +1)\\ \psi(N) \end{pmatrix} = A_N(x,\omega,E) \begin{pmatrix} \psi(1) \\ \psi(0)\end{pmatrix}.$$
Moreover, the transfer matrices is a classic example of an iterate of a quasiperiodic cocycle, by taking $A(x,\omega) = M_1(x,\omega,E).$

Now observe that any $M(2,\C)$ cocycle $A_N(x)$ for which $\det(A(x))$ is not identically zero can be renormalized to form an $SL(2,\C)$ cocycle (see e.g.\cite{JitomirskayaKosloverSchulteis}), however the resulting cocycle will lose pointwise boundedness if $\det(A(x))$ has zeros. This is precisely the nature of difficulty when extending $SL(2,\C)$ results to the $M(2,\C)$ case.

The (upper) Lyapunov exponent is defined as
\begin{equation}
L'(A,\omega) = \frac 1 N \int_{\T^d} \ln \norm{A_N(x,\omega)}dx.
\end{equation}
Note that, while $L'(A, \omega)$ need not be non-negative, the related object
\begin{equation}
L(A,\omega) = \lim_{N\to\infty} \int_{\T^d} L_N(\tilde A,\omega,x) dx,
\end{equation}
is, where $\tilde A \in SL(2,\C)$ is a renormalization of $A:$
\begin{equation}
\tilde A = \frac 1 {|\det A|^{1/2}} A.
\end{equation}
Moreover, $L_N$ and $L'_N$ are related by the following relation:
\begin{equation}
L_N(A,\omega) = L'_N(A,\omega) - \frac{1}{2} \int_{\T^d} \ln|\det(A(x))| dx.
\end{equation}
It follows that, when $\ln|\det(A(x))| \in L^1,$ both $L$ and $L'$ share the same regularity properties. 

It is often easier to deal with $L'$ when proving general boundedness and finite-scale continuity (see Sections \ref{Section:FirstSection} and \ref{Section:FiniteScaleCont}), but it is easier to deal with a non-negative quantity when proving and using our induction scheme (see Sections \ref{Section:Firstbasestep} and \ref{Section:Cont}). Thus both $L$ and $L'$ play a role in this paper.

Returning to \eqref{eq:JacobiOp}, an object related to the Lyapunov exponent for such operators is the integrated density of states (IDS), which maybe generally defined as in \cite{BourgainKlein}. Let $E_1 < E_2$ and define
$$k(x,\omega, E_1, E_2) = \limsup_{N \to \infty} \frac{1}{2N + 1} \set{\text{eigenvalues of} R_{[-N, N]} H_{x,\omega} R_{[-N,N]} \text{ in } [E_1, E_2]},$$
where $R_{[-N, N]}$ denotes the projection onto the interval $[-N, N].$ In our setting, this $\limsup$ is constant for Lebesgue a.e. $x.$ 
Clearly one may ask about the regularity of this object in the various parameters, $V, E_1, E_2, \omega.$ We obtain such a statement as a consequence of our main theorem (see Corollary \ref{Cor:IDS}).


The continuity of the Lyapunov exponent (in both cocycle and frequency) in this setting has been studied extensively in \cite{PowellContinuity}, where the author adapted an argument of Bourgain originally used to study $SL(2,\C)$ cocycles to study singular cocycles. 

The (pointwise) modulus of continuity of the Lyapunov exponent for one-frequency quasi-periodic cocycles has been studied by many authors.
While the Lyapunov exponent is known to be continuous in a very general settings, the (pointwise) modulus of continuity is a more delicate matter. This question was first studied in \cite{GoldsteinSchlagAnnals} for one-frequency (the underlying torus is $\T^1$) Schr\"odinger operators with fixed strongly diophantine frequency $\omega.$ There, authors also obtained analogous modulus of continuity for the IDS in the same setting. A key component of this proof was an a priori positivity assumption: $L(E) > 0.$ Shortly afterwards, the Lyapunov exponent for Schr\"odinger cocycles was shown to be continuous in both energy (for all frequencies) and frequency (at irrational frequencies) without positivity or Diophantine assumptions \cite{BourgainJitomirskayaContinuity}. This work was unable to obtain a modulus of continuity, however. It turns out that any modulus of continuity better than $\log$-H\"older requires both positivity and some arithmetic (perhaps weak) assumption \cite{BourgainBook}. 

Many authors have worked to extend various results from \cite{GoldsteinSchlagAnnals} with some success. The modulus of continuity of the Lyapunov exponent in the one-frequency setting has been extended to singular cocycles and fixed weaker diophantine frequencies \cite{MR3177775}; the multifrequency case has proved more delicate, and the known results still require some restrictive condition on the frequency \cite{DuarteKleinBook}. The modulus for both the Lyapunov exponent and IDS has also been obtained and/or sharpened in a variety of settings closely related to the Schr\"odinger case \cite{MR3413976, MR1774640, MR3936094, MR3910418, MR3262622, MR4097552, MR4426315, MR2481750, MR4068814, MR3963928, MR4102993, MR3227161}


We prove the following.

\begin{mythm}\label{Thm:MainThm}
Let $A(x) \in C_\rho$ be an analytic quasi-periodic $M(2,\C)$-cocycle on $\T^d$ with a plurisubharmonic extension to the strip $|\Im(z_j)| < \rho.$ Suppose that $\omega \in \T^d$ is such that, for some $\sigma \geq 1,$ $\norm{k\cdot \omega} > \tau |k|^{-\sigma} > 0$ for all $k\in\Z^d$ with $0 < |k|,$ and suppose that $L(A,\omega) > 0.$ Then there is $\delta = \delta(A)$ and $\gamma = \gamma(d,\sigma) \leq 1,$ such that, for every $A' \in C_\rho$ with $\norm{A - A'}_\rho + \norm{\omega - \omega'} < \delta,$
\begin{equation}\label{eq:WeakHolder}
|L(A,\omega) - L(A',\omega')| < C(A) \exp\set{-c(A)\left(-\ln\left(\norm{A - A'}_\rho + \norm{\omega - \omega'}\right)\right)^{\gamma}}.
\end{equation}
\end{mythm}

We would like to make a few remarks at this point.

\begin{Remark}
Generally, a pointwise modulus of continuity of the form \eqref{eq:WeakHolder} is called pointwise $\gamma$-weak-H\"older continuity.
\end{Remark}

\begin{Remark}
The exponent $\gamma$ above may be computed explicitly in terms of the Diophantine parameter $\sigma$ and an exponent from a large deviation estimate, which depends only on the dimension, $d,$ of the torus $\T^d.$ In particular, improvement of the large deviation estimate would lead to improvement of the exponent $\gamma,$ which has been remarked in both \cite{GoldsteinSchlagAnnals} and \cite{DuarteKleinBook}.
\end{Remark}

\begin{Remark}
Notice that, while $\omega$ must satisfy $\omega \in \T^d$ is such that $\norm{k\cdot \omega} > \tau |k|^{-\sigma} > 0$ for all $k\in\Z^d$ with $0 < |k|,$ the frequency $\omega'$ need not. This is a key improvement on \cite{MR4102993}, where $\omega'$ had to also satisfy the same condition, and is a consequence of our argument from \cite{PowellContinuity}.
\end{Remark}

In the remainder of this paper, when there can be no ambiguity, we will write $\norm{A}$ in place of $\norm{A(x)} = \norm{A(x)}_\rho.$

As a corollary, we obtain analogous continuity for the IDS.

\begin{mycor}\label{Cor:IDS}
Let $H_{x,\alpha}$ be as in \eqref{eq:JacobiOp}, with $V \in C^\omega(\T^d, \R)$ a not identically singular analytic function on $\T^d.$  Suppose that $\alpha \in \T^d$ is such that $\norm{k\cdot \alpha} > \tau |k|^{-\sigma} > 0$ for all $k\in\Z^d$ with $0 < |k|,$ and suppose $L(E) > 0.$ Let $k(x, \alpha, E_1, E_2)$ be defined as above. Then for a.e. $x,$ $k(\alpha, E_1, E_2)$ obeys
\begin{align*}
|k(\alpha, E_1, E_2)| 
< \exp\set{-\left(-\ln\left(|E_1 - E_2|\right)\right)^{\gamma}}.
\end{align*}



\end{mycor}

This may be seen either as a consequence of Theorem \ref{Thm:LDT} via an argument of Bourgain \cite{BourgainIDS} (and extended first by Schlag \cite{MR1860759} and later by Liu \cite{WencaiDisc} for any Schr\"odinger operator with a large deviation estimate), or as a consequence of Theorem \ref{Thm:MainThm} by an argument of Goldstein-Schlag using the Thouless formula, which generalizes to multifrequency Jacobi operators. As these arguments are standard and not new, we refer readers to those works for details.

The rest of this paper is organized as follows. In Section \ref{Section:FirstSection}, we recall a few well-known results for subharmonic functions and the finite-scale Lyapunov exponents, which we use throughout. Then we prove a modulus of continuity for the finite-scale Lyapunov exponent in Section \ref{Section:FiniteScaleCont}. In Section \ref{Section:Firstbasestep}, we provide an inductive procedure to obtain a locally uniform rate of convergence for the Lyapunov exponent when the frequency is fixed. In Section \ref{Section:Cont}, we combine the uniform rate of convergence with modulus of continuity for the finite scale Lyapunov exponents to obtain corresponding modulus of continuity for the Lyapunov exponent in both frequency and cocycle.

\section{Preliminaries: uniform bounds on Lyapunov exponents, large deviations, and Avalanche Principle}\label{Section:FirstSection}

Throughout this paper, we use $C$ and $C(\cdot)$ to denote large constants which depend on uniform measurements of the cocycle and dimension, $c$ and $c(\cdot)$ to denote small constants which depend on uniform measurements of the cocycle and dimension, and $C_\rho$ to denote constants which depend on the parameter $\rho.$ Unless otherwise stated, these constants may change by multiplicative constants throughout their appearance in a proof, but will remain finite, non-zero, and uniform in their respective parameters. 

This section is devoted to a few essential preliminary results which will be used in later sections. They may be found in a variety of other papers, and are provided here without proof except where the proof introduces ideas useful in the study of singular cocycles. 

We begin with a uniform version of the Lojasiewicz inequality, which is used throughout and may be of independent interest to readers.

\begin{mylemma}[\cite{DuarteKleinBook} Lemma 6.1]\label{Lem:UnifLoj} Let $f(x) \in C_\rho(\T^d,\C)$ be such that $f(x)$ is not identically zero. Then there are constants $\delta = \delta(f) > 0, S = S(f) < \infty,$ and $b = b(f) > 0$ such that if $g(x) \in C_\rho(\T^d,\C)$ with $\norm{g - f}_\rho < \delta,$ then 
\begin{equation}
\left|\set{x \in \T^d: |g(x)| < t}\right| < S t^b
\end{equation}
for all $t > 0.$
\end{mylemma}

This has, as a consequence, uniform $L^2$-boundedness of $L_N(A,\omega)$ in $N, \omega,$ and locally in $A.$ This result may also be found in \cite{DuarteKleinBook}, but our proof differs from the one therein.

\begin{mylemma}
Let $(A,\omega)$ be an analytic $M(2,\C)$ cocycle for which $\det A(x)$ is not identically zero. There is $\delta = \delta(A) > 0$ and $C = C(A)$ such that, for every cocycle $(B,\omega')$ such that $\norm{A - B}_\rho + \norm{\omega - \omega'} < \delta,$ 
$$\norm{L_N(B,\omega')}_{L^2} < C(A).$$
\end{mylemma}

\begin{proof}
Note that it suffices to prove that $L_N'(B,\omega', x)$ and $\ln|\det B(x)|$ obey this $L^2$-estimate.  Clearly we have
\begin{align*}
\int_{\T^d} \ln |\det B(x)| & = \int_{|\det B(x)| \geq 1} + \sum_{j = 0}^\infty \int_{ 2^{-j} > |\det B(x)| \geq  2^{-j - 1}}.
\end{align*}
Moreover, since $A(x)$ is analytic, there is some $C(A) < \infty$ such that $\norm{A(x)}_\rho< C.$ Thus, for any $B$ such that $\norm{A - B}_\rho < C/2,$ we have
$$\norm{B(x)}_\rho < 2C.$$
It follows that
$$|\det B(x)| < 4C^2.$$
On $\set{x: \epsilon \leq |\det B(x)| < 1},$
$$\left|\ln |\det B(x)|\right| < -\ln \epsilon.$$
Hence, on $\set{x: 2^{- j - 1} \leq |\det B(x)| < 2^{-j}},$ 
$$\left|\ln |\det B(x)|\right| < (j + 1)\ln 2.$$
By Lemma \ref{Lem:UnifLoj}, we have, for $\norm{A - B}_\rho < \delta(A)$
$$\left|\set{x: 2^{- j - 1} \leq |\det B(x)| < 2^{-j}}\right| < S(A) \left(2^{- j}\right)^b.$$
Altogether, this yields
\begin{align*}
\int_{\T^d} |\ln |\det B(x)||^2 &\leq 4(\ln 4C(A))^2 + S(A) \sum 2^{-bj} \left((j + 1)\ln 2\right)^2\\
&\leq (\ln 2C(A))^2 + S(A)C(A)\\
&< C(A).
\end{align*}

Now consider $L_N'(B,\omega',x).$ We claim that, for some $C = C(A)$ and $\delta > 0,$ and every $B$ such that $\norm{A - B}_\rho < \delta,$ we have
$$|L_{N}'(B,\omega',x)| \leq C + \frac 1 N \sum_{j = 0}^{N - 1}\left|\ln|\det (B(x + j\omega'))|\right|.$$
The desired conclusion quickly follows from this inequality and the above estimate for $\ln|\det B(x)|,$ so we turn our attention to a proof.

Let $\delta > 0$ be such that Lemma \ref{Lem:UnifLoj} holds for the analytic functions $\det B(x)$ when $\norm{B - A}_\rho < \delta.$ Observe that there is $C = C(A) > 0$ such that $\norm{A}_\rho < e^C,$ so $\frac 1 N \ln\norm{A_N}_\rho \leq C.$ By our assumptions on $\norm{B - A}_\rho,$ it follows that
$$\norm{B}_\rho < e^{2C},$$
and thus
$$\frac 1 N \ln \norm{B_N(x)} < 2C.$$
Now we consider two sets:
$$F_+ := \set{x: \ln \norm{B_N(x)} \geq 0}$$
and
$$F_- := \set{x: \ln\norm{B_N(x)} < 0}.$$
Clearly, for $x \in F_+,$ the desired inequality holds. Consider $x \in F_-.$ Recall that, for any $2\times 2$ matrix $A,$ we have
$$|\det(A)| \leq 2\norm{A}^2,$$
so it follows that
$$\frac 1 N \ln\norm{B_N(x)} \geq \frac 1 {2N} \ln|\det(B_N(x))| = \frac 1 {2N} \sum_{j = 0}^{N - 1} \ln|\det(B(x + j\omega))|.$$
For $x \in F_-,$ this implies
$$|L_N'(B,\omega,x)| \leq \left|\frac 1 {2N} \sum_{j = 0}^{N - 1} \ln|\det(B(x + j\omega))|\right|,$$
which, after applying triangle inequality, is our desired bound.

\end{proof}

We conclude this section with two essential results which will be used in the sequel. The first is a so-called {\it{large deviation estimate}} and the second is a consequence of the Avalanche Principle which, in this form, was originally due to Bourgain for $SL(2,\C)$ cocycles, but extended to non-identically singular $M(2,\C)$ cocycles by us in \cite{PowellContinuity}. We refer readers to either \cite{BourgainContinuity} (for the $SL(2,\C)$ case) or \cite{PowellContinuity} (for the $M(2,\C)$ case) for proofs.

\begin{mythm}\label{Thm:LDT}[\cite{PowellContinuity} Theorem 2.2]
Let $(A,\omega)$ be an analytic $M(2,\C)$ cocycle, with an analytic extension to $|\Im(z_j)| < \rho,$ for which $\det A(x)$ is not identically zero.  Suppose $\omega \in \T^d$ is such that
$$\norm{k\cdot\omega} > \delta_0$$
for all $0 < |k| < K_0.$ Moreover, suppose 
$$N > K_0 \delta_0^{-1}.$$
Then there are $\delta = \delta(A) > 0,$ $c = c(d) < 1,$ and $C_\rho < \infty$ such that for any $B$ with with an analytic extension to $|\Im(z_j)| < \rho$ satisfying $\norm{B - A}_\rho < \delta,$ 
\begin{equation}
\left|\set{x \in \T^d: \left|L_N(B, x) - L_N(B) \right| > \rho^{-1} K_0^{-c}}\right| < e^{-C_\rho K_0^c}.
\end{equation}
\end{mythm}

\begin{mythm}\label{Thm:APApp}[\cite{PowellContinuity} Theorem 4.3]
Let $(A,\omega)$ be an analytic $M(2,\C)$ cocycle with an analytic extension to $|\Im(z_j)| < \rho.$ Fix $x\in \T^d$ and $\delta > 0.$ Let $N > N_0(\delta)$ be sufficiently large and $N | N_1$ with $N \leq N_1.$ Suppose that
\begin{align}
L_N(A, x)&> \delta\\
\left|L_N(A, x) - L_{2N}(A, x)\right| &< \frac{1}{100} L_N(A, x)\\
\max_{n = N, 2N}\left|L_n(A, x) - L_n(A, x + j N\omega)\right| &< \frac\delta{100},
\end{align}
for all $ j \leq N_1/N.$
Then 
\begin{align}\begin{split}
\bigg|L_{N_1}(A, x) + \frac{1}{n}\sum_{j = 0}^{n-1} L_N(A, x + &j N\omega) - \frac 2 n \sum_{j = 0}^{n - 1}L_{2N}(A, x + jN \omega)\bigg| \\
&< \exp\left(-\frac{N}{4}L_{N}(x)\right) + C L_N(A,x) \frac{N}{N_1}.
\end{split}\end{align}
Here $C$ is an absolute constant.
\end{mythm}

\begin{Remark}
In the above, $N_0$ is such that $N_0 \delta > 2.$
\end{Remark}


\section{Finite-scale weak-H\"older continuity}\label{Section:FiniteScaleCont}

In the Schr\"odinger cocycle case (and the $SL(2,\C)$ case more generally), one of the key observations is that, for fixed $N,$ $L_N(A,\omega)$ is jointly continuous in $A$ and $\omega$ for any $\omega.$ This is a simple consequence of the everywhere invertibility of the cocycle, $A.$ In this section, we prove an analogous result for non-identically-singular cocycles. Indeed, our strategy hinges on this fact, as we observe that, given a sequence of continuous functions (in this case, $L_N(A,\omega)$), the continuity of the limiting object (in this case, $L(A,\omega)$) may be obtained by quantitatively estimating the rate of convergence ($L_N \to L$) uniformly in $A$ and $\omega.$ This argument does not work if $L_N(A,\omega)$ is not continuous!

We do not claim that this result is novel; it is included for completeness, and because the proof introduces some ideas which are useful for studying singular cocycles.

\begin{mythm}\label{Thm:FiniteCont}
Let $A(x)$ be an analytic quasi-periodic $M(2,\C)$-cocycle on $\T^d$ with an analytic extension to the strip $|\Im(z_j)| < \rho.$ 
There are $\delta_1(A) > 0$ and $b(A) > 0$ such that for any $\alpha \geq 1,$ whenever $\norm{A - B} < \delta_1,$ we have
\begin{align}\begin{split}
|L_N'(A,\omega) - L_N'(B,\omega')| &< C(A)e^{-2N^{\alpha - 1}b(A)} \\
&\quad+ e^{2N^\alpha} \left(\norm{A - B} + N\norm{\omega - \omega'}\right).\end{split}
\end{align}
Moreover,
\begin{align}
\left|\int_{\T^d} \ln|\det A(x)| - \ln|\det B(x)| dx \right| &< 
C(A) \norm{A - B}^{b/(1 + b)}.
\end{align}

\end{mythm}

\begin{Remark}
Note that this implies, by the definition of $L_N$ and $L_N',$ that
\begin{align}\begin{split}
|L_N(A,\omega) - L_N(B,\omega')| &< C(A)e^{-2N^{\alpha - 1}b(A)} \\
&\quad+ e^{N^\alpha} \left(\norm{A - B} + N\norm{\omega - \omega'}\right).\end{split}
\end{align}
\end{Remark}

\begin{proof}
We first prove the result for $L_N'(A).$ The result for the determinant follows from an analogous argument, and we will discuss how to adjust the proof afterwards.

Let $C_0(A)$ be such that $\norm{A(x)} \leq e^{C_0}.$ For $\norm{A - B}$ sufficiently small, depending only on $C_0,$ $\norm{B(x)} \leq e^{2C_0}.$
Fix $\delta_0 > 0$ and set
\begin{align}
F_A(\delta_0) &= \set{x\in \T^d: \norm{A_N} < \delta_0}\\
F_B(\delta_0) &= \set{x\in \T^d: \norm{B_N} < \delta_0}.
\end{align}

Clearly 
$$\int_{\T^d} = \int_{F_A\cap F_B} + \int_{F_A^c \cap F_B} + \int_{F_A \cap F_B^c} + \int_{F_A^c \cap F_B^c}.$$

Note that $|\det(A_N(x))| \leq \norm{A_N(x)}^2$ for all $x.$ Thus
$$F_A(\delta_0) \subset \set{x: |\det A_N(x)| < \delta_0^2} \subset \bigcup_{j = 0}^{N - 1} \set{x: |\det A(x _ j\omega)| < \delta_0^{2/N}},$$
and
$$F_b(\delta_0) \subset \set{x: |\det B_N(x)| < \delta_0^2} \subset \bigcup_{j = 0}^{N - 1} \set{x: |\det B(x + j\omega')| < \delta_0^{2/N}}.$$
Moreover, we know $\ln\norm{A_N(x)} \in L^p, 1 \leq p < \infty$ and, for $\norm{A - B} < \delta_1 = \delta_1(A),$ $\ln\norm{B_N(x)}\in L^p, 1 \leq p < \infty$ and $\norm{\ln\norm{B(x)}}_p < C(p) \norm{\ln\norm{A(x)}}_p.$ These all hold uniformly in $N.$
This implies that, for $\norm{A - B} < \delta_1,$
$$\int_{F_A\cap F_B} + \int_{F_A^c \cap F_B} + \int_{F_A \cap F_B^c} \leq 2C(A) |F_A(\delta_0)| + C(A) |F_B(\delta_0)|,$$
where $|F_B|$ denotes the Lebesgue measure.
Moreover, by the Lojaciewz inequality, 
$$|F_A(\delta_0)| < C(A) N \delta_0^{2b(A)/N}$$
and
$$|F_B(\delta_0)| < C(A) N \delta_0^{2b(A)/N}.$$

Finally, for $x \in F_A^c \cap F_B^c,$ we have
\begin{align}
|\ln\norm{A_N(x)} - \ln\norm{B_N(x)}| &\leq \max\set{\norm{A_N - B_N}\norm{B_N}^{-1}, \norm{B_N - A_N}\norm{A_N}^{-1}}\\
&\leq \norm{A_N - B_N}\delta_0^{-1}.
\end{align}

Thus, we have
\begin{align}\begin{split}
\int_{\T^d} &= \int_{F_A\cap F_B} + \int_{F_A^c \cap F_B} + \int_{F_A \cap F_B^c} + \int_{F_A^c \cap F_B^c}\\
&\leq C(A) \delta_0^{2b(A)/N} + N^{-1}\norm{A_N - B_N}\delta_0^{-1}.\end{split}
\end{align}
Moreover, a standard telescoping argument implies
\begin{align}
\norm{A_N - B_N} &\leq N \norm{A}^{2N}\left(\norm{A - B} + \norm{N\omega - N\omega'}\right)\\
&\leq N \norm{A}^{2N}\left(\norm{A - B} + N\norm{\omega - \omega'}\right).
\end{align}
We thus have
\begin{align}\begin{split}
|L_N'(A,\omega) - L_N'(B,\omega')| &< C(A) \delta_0^{2b(A)/N} \\
&\quad+ \norm{A}^{2N}\left(\norm{A - B} + N\norm{\omega - \omega'}\right)\delta_0^{-1}.\end{split}
\end{align}
Finally, set 
$$\delta_0 = e^{-N^\alpha}, \quad \alpha > 1.$$
Recall that $\norm{A} \leq e^{C_0},$ and consequently
\begin{align*}
|L_N'(A,\omega) - L_N'(B,\omega')| &< C(A) e^{-2N^{\alpha - 1}b(A)} \\
&\quad+ e^{2C_0 N} \left(\norm{A - B} + N\norm{\omega - \omega'}\right)e^{N^\alpha}\\
&\leq C(A)e^{-2N^{\alpha - 1}b(A)} + e^{2N^\alpha} \left(\norm{A - B} + N\norm{\omega - \omega'}\right).
\end{align*}

To obtain the continuity for $\ln|\det(A)|,$ we use the above argument where $F_A(\delta_0)$ and $F_B(\delta_0)$ are redefined as
\begin{align}
F_A(\delta_0)' &:= \set{x: |\det(A(x))| < \delta_0}\\
F_B(\delta_0)' &:= \set{x: |\det(B(x))| < \delta_0}.
\end{align}
At the end, we take $\delta_0 = \norm{A - B}^{1/(1 + b)},$ where $b = b(a)$ is as in the Lojasiewicz inequality.
\end{proof}


\section{Lyapunov exponents: local uniform rate of convergence with fixed Diophantine frequency}\label{Section:Firstbasestep}
In this section, we will establish an inductive scheme to obtain a local uniform rate of convergence for the Lyapunov exponents corresponding to fixed Diophantine frequencies. We will break this up into multiple steps.

\begin{mylemma}
Let $A(x)$ be an analytic quasi-periodic $M(2,\C)$-cocycle on $\T^d$ with with an analytic extension to $|\Im(z_j)| < \rho$ such that $\det A(x) \not\equiv 0.$ Suppose that $\omega \in \T^d$ is such that $\norm{k\cdot \omega} > \delta_0 > 0$ for all $k\in\Z^d$ with $0 < |k| \leq K_0,$ where $K_0$ satisfies
\begin{align}
K_0 &> (\rho \kappa)^{-C}\label{eq:Mixed:hyp2}\\
N_0 &> \kappa^{-C}\delta_0^{-1}K_0.\label{eq:Mixed:hyp3}
\end{align}
Moreover, suppose that 
\begin{align}
L_{N_0}(A) &> 10^3\kappa\\
|L_{N_0} - L_{2N_0}| &< \frac1{100} L_{N_0}.\label{eq:APHardCond}
\end{align} 
Then 
\begin{equation}
|L_{N} + L_{N_0} - 2L_{2N_0}| < C'(A)e^{-C_\rho K_0^c}
\end{equation}
when $N_0 | N, N = 2^jN_0$ for some $j \geq 0,$ and 
\begin{equation}
N < N_0 e^{\frac 12 \rho K_0^c}.
\end{equation} 
In the above, $C$ is a sufficiently large absolute constant depending only on $d$ which is defined in the proof, $c = C^{-1},$ and $C'(A)$ is a constant which is uniform in $A.$
\end{mylemma}


\begin{proof}
First, consider the set
$$G = \set{x: |L_{N_0}(A,x) - L_{N_0}(A)| > \kappa}.$$
Our large deviation theorem yields
$$|G| < e^{-C_\rho K_0^{c}}.$$ 
Thus, on $\T^d\backslash G,$
$$L_{N_0}(A,x) > 999\kappa,$$
and 
\EQ{\left|\int_G L_{N_1}(x) + L_{N_0}(x) - 2L_{2N_0}(x) dx\right| < e^{-C_\rho K_0^{c}}.}
Let us now consider a further restriction to the set, $F \subset \T^d\backslash G,$ consisting all of $x$ such that $$|L_{N_0}(A,x) - L_{N_0}(A, x + jN_0\omega)| < 2\kappa$$
for $1 \leq j \leq N_1/N_0.$ Another application of our large deviation theorem implies this set satisfies $$|\T^d \backslash F| < \frac{N_1}{N_0} e^{-C_\rho K_0^c} < e^{-C_\rho K_0^c}$$
for $N_1 < N_0e^{C_\rho K_0^c}.$ 

Finally, one last application of our large deviation theorem, applied to \eqref{eq:APHardCond}, implies the set
$$H = \set{x \in \T^d \backslash G: |L_{N_0}(x) - L_{2N_0}(x)| < \frac1{100} L_{N_0}(x)}$$
obeys the measure estimate
$$|\T^d\backslash H| < e^{-C_\rho K_0^c}.$$  

Now we see that every $x\in F \cap H$ satisfies the three hypotheses of Theorem \ref{Thm:APApp}, so on $F\cap H$ we obtain
\begin{align}\begin{split}\label{eq:FirstAPAppScale}
\bigg|L_{N_1}(A, x) + \frac{1}{n}\sum_{j = 0}^{n-1} L_{N_0}(&A, x + j N_0\omega) - \frac 2 n \sum_{j = 0}^{n - 1}L_{2N_0}(A, x + jN_0 \omega)\bigg| \\
&< \exp\left(-\frac{N_0}{4}L_{N_0}(x)\right) + C' L_{N_0}(A,x)\frac{N_0}{N_1}.
\end{split}\end{align}

We may now integrate the left hand side of the above inequality, over $\T^d:$
\EQ{\int_{\T^d} = \int_{F\cap H} + \int_{F^c \cap H} + \int_{F\cap H^c}.}
Using $L_{N_0}(x) > 999\kappa$ on $F\cap G,$ our restriction on $N_1,$ and \eqref{eq:FirstAPAppScale} we have
\EQ{\int_{F\cap H} < e^{-\frac{999}{4}N_0\kappa} + C'(A) L_{N_0}(A) e^{-C_\rho K_0^c}.}
By the uniform $L^2$-boundedness of the integrand (i.e. $\norm{L_N(A)}_2 < C(A)$) and our measure estimates on $|\T^d \backslash F|$ and $|\T^d\backslash H|,$ we have
\begin{align}
\int_{F^c\cap H} &\leq C'(A) e^{-C_\rho K_0^c}\\
\int_{F\cap H^c} &\leq C'(A) e^{-C_\rho K_0^c}.
\end{align}
Combining everything yields our result.

\end{proof}

Next, we show that $L(A) > 0$ ensures that $|L_{N_0} - L_{2N_0}|$ is always small enough.

\begin{mylemma}\label{Lemma:BaseStep}
Let $A(x)$ be an analytic quasi-periodic $M(2,\C)$-cocycle on $\T^d$ with an analytic extension to the strip $|\Im(z_j)| < \rho.$ Suppose that $\omega \in \T^d$ is such that $\norm{k\cdot \omega} > \delta_0 > 0$ for all $k\in\Z^d$ with $0 < |k| \leq K_0,$ where $K_0$ satisfies
\begin{align}
K_0 &> (\rho^{1 + c}\kappa)^{-C}\label{eq:Mixed:hyp2}\\
N_0 &> \kappa^{-C}\delta_0^{-1}K_0.\label{eq:Mixed:hyp3}
\end{align}
Moreover, suppose that $N_0 = a 2^b,$ for some $a \in \N$ and $b > -C\ln\kappa$ (here $C$ is the same as above) and
\EQ{L(A) > 10^3\kappa.} 
Then 
\begin{equation}
|L_{N_0} - L_{2N_0}| < \frac{1}{100} L_{N_0}.
\end{equation}
In the above, $C$ is a sufficiently large absolute constant depending only on $d$ which is defined in the proof and $c = C^{-1}.$
\end{mylemma}

\begin{Remark}
The condition $N_0 = a 2^b$ is not, in fact, necessary, but allowing for more general scales requires a technical approximation argument. All we will need for applications is $N_0 = a 2^b,$ so that is what we will prove. We refer readers to \cite{PowellContinuity} Section 5 for a proof for arbitrary $N_0$ large.
\end{Remark}

\begin{proof}
First, as before, we remark that large deviation implies
\EQ{
L_{N_0}(x) > L_{N_0} - \kappa \geq L - \kappa > 999\kappa
}
away from a set, call it $G,$ of measure $|G| < e^{-C_\rho K_0^c}.$ Since $L_{N_0}(x) \in L^2,$ we have 
\EQ{\left|\int_{\T^d \cap G} L_{N_0}(x) - L_{2N_0}(x) dx\right| < Ce^{-C_\rho K_0^c}.} 
Now consider the set $F$ consisting of all $x\in \T^d\backslash G$ such that
\begin{equation}
|L_N(x) - L_N(x + j\omega)| < \kappa
\end{equation}
for all $N = 2^{-s}N_0,$ with $0 \leq s \leq s_0 := - C_1\ln\kappa,$ and $j \leq 2N_0$ such that $2^{-s}N_0 | j$ for some $j.$

Note that $K_0^{-c} < \rho^{1 + c} \kappa,$ by our assumption on $K_0.$ Thus, for all such $N,$ 
$$\set{x: |L_N(x) - L_N| > \kappa} \subset \set{x: |L_N(x) - L_N| > \rho^{-1 - c}K_0^{-c}}.$$
Moreover, 
\begin{align}
2^{-s}N_0 &> 2^{-s} \kappa^{-C} \delta_0 K_0 \\
&> \kappa^{C_1 - C} \delta_0 K_0.
\end{align}
Taking $C \geq C_1$ ensures that the right hand side is at least $\delta_0 K_0.$ All of this together implies that we may apply our large deviation to control the measure of $\T^d\backslash F$ as follows.

Large deviation, and our condition on $K_0$ implies
\begin{align}
\left|\T^d\backslash F\right| &< \sum_{s = 0}^{s_0} \sum_j \left|\set{|L_N(x) - L_N(x + j\omega)| > \kappa}\right|\\
&\leq \sum_s \frac{2N_0 2^s}{N_0} e^{-C_\rho K_0^c}\\
&\leq (2^{s_0} - 1)\frac{2N_0}{N_0} e^{-C_\rho K_0^c}\\
&\leq \kappa^{-C_1} \frac{2N_0}{N_0} e^{-C_\rho K_0^c}\\
&\leq K_0^{C_1/C} \frac{2N_0}{N_0} e^{-C_\rho K_0^c}\\
&\leq \frac{2N_0}{N_0} e^{-C_\rho K_0^c}.
\end{align}
Thus we have
$$\left|\T^d\backslash F\right| < e^{-C_\rho K_0^c} < \kappa.$$
Throughout this computation, we needed to take $K_0$ sufficiently large, depending only on $\rho, C,$ and $C_1.$ Throughout, $\rho$ is fixed, and we will later fix $C$ and $C_1$ in the proof independent of all other parameters, so this will not pose an obstacle.

Since $L_N(A,x)$ is uniformly $L^2$-bounded in $N,$ 
$$\int_{\T^d\backslash F} \left| L_{N_0}(x) - L_{2N_0}(x)\right| dx < C(A) e^{-C_\rho K_0^{c}} < C(A) \kappa.$$

Let us now define 
$$J := \set{x \in F: L_{2^{-s_0}N_0}(x) < \kappa^{-1}}.$$
By Chebyschev's inequality, $|\T^d \backslash J| < C(A) \kappa,$ and thus
$$\int_{\T^d\backslash J} \left| L_{N_0}(x) - L_{2N_0}(x)\right| dx < C(A) \kappa.$$

Now consider $x \in J.$
By definition, $L_{N_0}(x) > 999 \kappa.$ Define scales $N_{0s}$ inductively by $N_{01} = 2^{-1} N_0$ and $N_{0(s + 1)} = 2^{-1} N_{0s}.$ Since $x \in F,$ we have
\begin{align}
999\kappa N_0 \leq N_{0} L_{N_{0}}(x) &\leq \sum_{j = 0}^{N_0/N_{0s}} N_{0s} L_{N_{0s}}(x + j N_{0s}\omega) \\
&\leq \sum_{j = 0}^{N_0/N_{0s}} (N_{0s} L_{N_{0s}}(x) + N_{0s}\kappa)\\
&= N_0 L_{N_{0s}}(x) + N_0 \kappa.
\end{align}
By a similar computation, we have, for $1 \leq s < s' \leq s_0$
\begin{equation}
L_{N_{0s}}(x) \leq L_{N_{0s'}}(x) + \kappa.
\end{equation}
For each such $x,$ we define $s(x)$ and new length scales $N_{00}(x) = N_{0s(x)} = 2^{-s(x)} N_0$ such that 
\begin{align}
N_{00}(x) &\leq \kappa N_0\\
|L_{N_{00}(x)}(x) - L_{2N_{00}(x)}| &< \frac{1}{100} L_{N_{00}(x)}.
\end{align}
The first condition is achievable as long as $C_1 > 1.$ The second condition may be rewritten as
$$\frac{99}{100} L_{N_{00}(x)}(x) < L_{2N_{00}(x)} < \frac{101}{100} L_{N_{00}(x)}(x).$$
The right inequality holds by $L_{2N_{00}(x)} < L_{N_{00}(x)} + \kappa.$ The left inequality holds by taking $C_1$ sufficiently large (say $C_1 > 5$) and using $L_{2^{-s_0}N_0}(x) < \kappa^{-1}.$

The hypotheses of Theorem \ref{Thm:APApp} are now satisfied for $N_{00}(x)$ with $\delta = 100 \kappa,$ so we have
\begin{align}
\begin{split}
\bigg|L_{N}(x) &+ \frac{N_{00}}{N} \sum_{j = 0}^{\frac{N}{N_{00}} - 1}L_{N_{00}(x)}(x + jN_{00}\omega) \\
&\quad - \frac{2N_{00}}{N} \sum_{j = 0}^{\frac{N}{N_{00}}-1} L_{2N_{00}(x)}(x+ jN_{00}\omega)\bigg| \\
&< e^{-cN_{00}L_{N_{00}}(x)} + C L_{N_{00}}(A,x)\frac{N_{00}(x)}{N}
\end{split}
\end{align}
for $N = N_0, 2N_0.$ Since $x \in J,$ we have $L_{N_{00}}(x + jN_{00}\omega) = L_{N_{00}}(x) + O(\kappa)$ and $L_{N_{00}}(x) > 998\kappa.$ Moreover, $L_{N_{00}}(A,x) < L_{2^{-s_0}N_0}(A,x).$ Consequently,
$$|L_{N_0}(x) - L_{2N_0}(x)| < 10\kappa + C(A)\frac{N_{00}(x)}{N_0} + C(A) \frac{N_{00}(x)}{2N_0}.$$
Since $N_{00}(x) < \kappa N_0,$ we have
$$|L_{N_0}(x) - L_{2N_0}(x)| < C(A) \kappa.$$
Hence, for $x \in F,$ we must have
$$|L_{N_0}(x) - L_{2N_0}(x)| < \frac 1{200} L_{N_0}(x).$$
It now follows that
\EQ{
|L_{N_0} - L_{2N_0}| < \frac 1{200} L_{N_0} + C(A)e^{-C_\rho K_0^c} < \frac 1{100} L_{N_0}.}
\end{proof}

The previous two lemmas immediately imply the following.
\begin{mythm}\label{Thm:BaseStep}
Let $A(x)$ be an analytic quasi-periodic $M(2,\C)$-cocycle on $\T^d$ with a plurisubharmonic extension to the strip $|\Im(z_j)| < \rho$ such that $\det A(x) \not\equiv 0.$ Suppose that $\omega \in \T^d$ is such that $\norm{k\cdot \omega} > \delta_0 > 0$ for all $k\in\Z^d$ with $0 < |k| \leq K_0,$ where $K_0$ satisfies
\begin{align}
K_0 &> (\rho^{1 + c}\kappa)^{-C}\label{eq:Mixed:hyp2}\\
N_0 &> \kappa^{-C}\delta_0^{-1}K_0.\label{eq:Mixed:hyp3}
\end{align}
Moreover, suppose that $N_0 = a 2^b,$ for some $a \in \N$ and $b > -C\ln\kappa$ (here $C$ is the same as above) and
\begin{align}
L(A) &> 10^3\kappa.
\end{align} 
Then 
\begin{equation}
|L_N + L_{N_0} - 2L_{N_0}| < C(A)e^{-C_\rho K_0^c}
\end{equation}
when $N_0 | N, N = 2^jN_0$ for some $j \geq 0,$ and 
\begin{equation}
N < N_0 e^{C_\rho K_0^c}.
\end{equation} 
\end{mythm}

Elementary iteration allows us to extend this result from $N < N_0e^{C_\rho K_0^c}$ to $N < \exp(\exp(C_\rho K_0^c)).$ This formulation will be used to establish weak-H\"older continuity in frequency.

\begin{mycor}\label{Cor:BaseStep}
Let $A(x)$ be an analytic quasi-periodic $M(2,\C)$-cocycle on $\T^d$ with a plurisubharmonic extension to the strip $|\Im(z_j)| < \rho.$ Suppose that $\omega \in \T^d$ is such that $\norm{k\cdot \omega} > \delta_0 > 0$ for all $k\in\Z^d$ with $0 < |k| \leq K_0,$ where $K_0$ satisfies
\begin{align}
K_0 &> (\rho^{1 + c}\kappa)^{-C}\label{eq:Mixed:hyp2}\\
N_0 &> \kappa^{-C}\delta_0^{-1}K_0.\label{eq:Mixed:hyp3}
\end{align}
Moreover, suppose that $N_0 = a 2^b,$ for some $a \in \N$ and $b > -C\ln\kappa$ (here $C$ is the same as above) and
\begin{align}
L(A) &> 10^3\kappa.
\end{align} 
Then 
\begin{equation}
|L_N + L_{N_0} - 2L_{N_0}| < C(A)e^{-\rho K_0^c}
\end{equation}
when $N_0 | N, N = 2^jN_0$ for some $j \geq 0,$ and 
\begin{equation}
N < \exp(\exp(\frac 12 \rho K_0^c)).
\end{equation} 
\end{mycor}

By imposing a Diophantine condition on $\omega,$ we are able to perform a more delicate iteration scheme to extend this result to the limit.

\begin{mythm}\label{Thm:QualIteration}
Let $A(x)$ be an analytic quasi-periodic $M(2,\C)$-cocycle on $\T^d$ with a plurisubharmonic extension to the strip $|\Im(z_j)| < \rho.$ Suppose that $\omega \in \T^d$ is such that $\norm{k\cdot \omega} > \tau |k|^{-\sigma} > 0$ for all $k\in\Z^d$ with $0 < |k|.$ Then for every $N_0 = a2^b > N',$ we have, for $\eta = \eta(d,\sigma),$
\begin{equation}
|L + L_{N_0} - 2L_{2N_0}| < C(A) \exp\set{-C_\rho N_0^{\eta}}.
\end{equation}
Here $C$ is the same sufficiently large absolute constant from Lemma \ref{Lemma:BaseStep} and $c = C^{-1}$ above. Here $N'$ is a constant defined in our proof such that $N_0$ satisfies the divisibility criterion needed to appeal to Lemma \ref{Lemma:BaseStep}.
\end{mythm}

\begin{proof}
For simplicity, we will assume $\rho = 1,$ though the argument and result hold for general $\rho,$ we would just have to adjust $N_s, K_s,$ and $\kappa_s$ accordingly. We will begin by applying Theorem \ref{Thm:BaseStep}. We will define the appropriate parameters first. Fix $c, C, C_1$ as in Theorem \ref{Thm:BaseStep}, and let $K = K(c, C, C_1, \sigma)$ be large enough such that
\begin{equation}
K^{C_1/C} e^{-K^{c}} < e^{-C_\rho K^c}
\end{equation}
and
\begin{equation}
e^{C_\rho K^{c/2}} > 2K^{1 + \sigma/2}.
\end{equation}
The first condition ensures that every $K_0 > K$ satisfies the largeness condition we imposed during the proof of Lemma \ref{Lemma:BaseStep}, while the second condition will ensure that our base step can be iterated (see below for details).

Take any $N_0 = 2^b > \tau^{-1} K^{\sigma + 2} = N'.$ We now define the rest of our base parameters. Define $K_0, \kappa_0,$ and $\delta_0$ via 
\begin{align}
N_0 &=: \tau^{-1} K_0^{\sigma + 2},\\
K \leq K_0 &=: \kappa_0^{-C},\\
\delta_0 &:= \tau K_0^{-\sigma}.
\end{align} 
Lemma \ref{Lemma:BaseStep} is applicable with these parameters, and we have
$$\left|L_{N_1} + L_{N_0} - 2L_{N_0}\right| < C(A)e^{-C_\rho K_0^c}$$
for $N_0 | N_1, N_1 = 2^j N_0 < N_0e^{C_\rho K_0^c}.$

Now define 
\begin{align}
\kappa_1 &= \kappa_0^2,\\
K_1 = \kappa_1^{-C} &= \kappa_0^{-2C} = K_0^2,
\end{align}
and 
\EQ{
\delta_1 = \tau K_1^{-\sigma}.
}
Let $N_1 = N_0 e^{\frac 12 K_0^c},$ where this is meant as: $N_1$ is the largest integer multiple of $N_0$ no greater than $N_0 e^{C_\rho K_0^c}.$ 
We clearly have $N_1 \geq \frac12 N_0 e^{C_\rho K_0^c}.$ Moreover, by our choice of $\kappa_1, K_1,$ and $\delta_1,$ we have
\begin{align}
N_1 &\geq \frac 12 \kappa_0^{-C(\sigma + 2)} e^{C_\rho K_1^{c/2}}\\
&= \frac 12 \kappa_1^{-C} \kappa_1^{-C\sigma/2} e^{C_\rho K_1^{c/2}}\\
&= \frac 12 \tau^{-1/2}\kappa_1^{-C} \delta_1^{-1/2} e^{C_\rho K_1^{c/2}}\\
&\geq  \kappa_1^{-C} \delta_1^{-1} K_1,
\end{align}
for $K_1$ sufficiently large. In fact, we just need $K_0$ sufficiently large so that $e^{C_\rho K_0^{c/2}} > 2 K_0^{1 + \sigma/2},$ which is certainly satisfied by our initial choice of $K.$ 

Consequently, Theorem \ref{Thm:BaseStep} is applicable with $\kappa_1, K_1, \delta_1,$ and $N_1$ as above:
\begin{align}
|L_{N_2} + L_{N_1} - 2L_{2N_1}| &< C(A) e^{-C_\rho K_1^c}\\
&= C(A) e^{-C_\rho K_0^{2c}},\\
\end{align}
for $N_1 | N_2, N_2 = 2^j N_1 < N_1 e^{C_\rho K_1^c}.$

Now we define 
\begin{align}
\kappa_2 &:= \kappa_1^2 = \kappa_0^4,\\
K_2 &:= \kappa_2^{C} = \kappa_0^{4C},\\
\delta_2 &:= \tau K_2^{-\sigma},
\end{align}
and
\EQ{N_2 := N_1 e^{\frac 12 K_1^c} \geq \frac 12 N_1 e^{C_\rho K_1^c}.}
Clearly, by our choice of $\kappa_2, K_2, \delta_2,$ and what we know of $N_1,$ we have
\begin{align}
N_2 &\geq \frac 12 \kappa_1^{-C} \delta_1^{-1} K_1 e^{C_\rho K_1^c}\\
&= \frac 12 \tau^{-1} \kappa_2^{-C/2} K_2^{\sigma/2 + 1/2} e^{C_\rho K_2^{c/2}}\\
&= \frac 12 \tau^{-1/2} \kappa_2^{-C/2} \delta_2^{-1/2} K_2^{1/2}e^{C_\rho K_2^{c/2}}\\
&\geq \kappa_2^{-C} \delta_2^{-1} K_2
\end{align}
whenever $K_2$ satisfies $2K_2^{1 + \sigma/2}  \leq e^{C_\rho K_2^{c/2}}.$ It is easy to see that our original condition on $K$ ensures that $K_2$ is large enough.

We may now apply Lemma \ref{Lemma:BaseStep} with $\kappa_2, K_2, \delta_2,$ and $N_2$ as above to obtain
\EQ{|L_{N_3} + L_{N_2} - 2L_{2N_2}| < C(A)e^{-C_\rho K_2^c} < C(A) e^{-C_\rho K_0^{4c}}.}

Continuing in this way, we inductively define 
\begin{align}
\kappa_s &:= \kappa_{s - 1}^2, \\
K_s &:= \kappa_s^{-C}, \\
\delta_s &:= \tau K_s^{-\sigma},
\end{align} 
and 
\EQ{N_s := N_{s - 1}e^{\frac 12 K_{s - 1}^c} \geq \frac 12 N_{s - 1}e^{C_\rho K_{s - 1}^c}.}
Moreover, $N_{s - 1} \geq \kappa_{s - 1}^{-C} \delta_{s - 1}^{-1} K_{s - 1},$ and $e^{C_\rho K_{s - 1}^{c/2}} > 2 \tau^{1/2} K_{s - 1}^{1 + C/2}.$ Following the same argument as before, we thus have
\begin{align}
N_s &\geq \frac 12 \kappa_{s - 1}^{-C} \delta_{s - 1}^{-1} K_{s - 1} e^{C_\rho K_{s - 1}^c}\\
&= \frac 12 \tau^{-1} \kappa_{s}^{-C/2} \delta_{s - 1}^{-1} K_{s - 1} e^{C_\rho K_{s - 1}^c}\\
&= \frac12 \tau^{-1/2} \kappa_s^{-C/2} \delta_s^{-1/2} K_{s}^{1/2} e^{C_\rho K_s^{c/2}} \\
&\geq \kappa_s^{-C} \delta_s^{-1} K_s.
\end{align}
It follows that 
\EQ{|L_{N_{s}} + L_{N_{s - 1}}- 2L_{2N_{s-1}}| < C(A) e^{-C_\rho K_0^{2^{s-1}c}}.}

At this point, we note that, at each step, we could have used $2N_s,$ rather than $N_s,$ and obtained identical inequalities. Combining this fact with the established inequalities, we obtain
\EQ{
\left|L_{N_s} + L_{N_0} - 2 L_{2N_0}\right| < C(A) \sum_{j = 0}^{s - 1}e^{-C_\rho K_0^{2^jc}}.
}
Altogether, this yields
\begin{align}
|L_{N_s} + L_{N_0} - 2L_{2N_0}| &< \sum_{j = 0}^{\infty} C(A) e^{-C_\rho K_0^{2^jc}} \\
&\leq 2C(A) e^{-C_\rho K_0^c}.
\end{align}
Taking $s \to \infty,$ we obtain
\begin{equation}
|L + L_{N_0} - 2L_{2N_0}| \leq 2 C(A) e^{-C_\rho K_0^c}.
\end{equation}
Now, using our initial choice of $K_0,$ we finally obtain
\EQ{
|L + L_{N_0} - 2L_{2N_0}| \leq C(A) e^{-C_\rho N_0^{c/(\sigma + 2)}}
}
as desired, with $\eta = c(d)/(\sigma + 2).$

\end{proof}

\section{Weak-H\:older continuity of the Lyapunov exponent}\label{Section:Cont}
We begin this section by combining Theorem \ref{Thm:QualIteration} and Theorem \ref{Thm:FiniteCont} to obtain weak-H\"older continuity in $A.$ Then we will show that weak-H\"older continuity in $\omega$ follows from Theorem \ref{Thm:QualIteration} and the continuity of the Lyapunov exponent in $\omega.$

\begin{mythm}\label{Thm:CocycleCont}
Let $A(x)$ be an analytic quasi-periodic $M(2,\C)$-cocycle on $\T^d$ with an analytic extension to the strip $|\Im(z_j)| < \rho.$ Suppose that $\omega \in \T^d$ is such that $\norm{k\cdot \omega} > \tau |k|^{-\sigma} > 0$ for all $k\in\Z^d$ with $0 < |k|.$ Then there is $\delta = \delta(A)$ and $\gamma = \gamma(d,\sigma) \leq 1,$ such that, for every $\norm{A - B} < \delta,$
\begin{equation}
|L(A,\omega) - L(B,\omega)| < C(A) \exp\left(-c(A)\left(-\ln\norm{A - B}\right)^{\gamma}\right).
\end{equation}
\end{mythm}

\begin{proof}
Since this result involves a fixed frequency, we will omit it in our notation. Consider $\delta(A) > 0$ such that Theorem \ref{Thm:FiniteCont} applies for $\norm{A - B} < \delta.$ Then we have, for every $N = a2^b > N',$ where $N'$ is the same absolute constant from Theorem \ref{Thm:QualIteration}, and every $\alpha > 1,$
\begin{align*}
|L(A) - L(B)| &\leq |L(A) + L_N(A) - 2L_{2N}(A)|\\
&\quad + |L_N(B) + L_N(A)| + 2|L_{2N}(B) - L_{2N}(A)|\\
&\quad + |L(B) + L_N(B) - 2L_{2N}(B)|\\
&\leq C(A)e^{-C_\rho N^\eta}\\
&\quad + C(A)e^{-2N^{\alpha - 1}b(A)} + e^{N^\alpha} \norm{A - B}\\
&\quad + C(B)e^{-2N^{\alpha - 1}b(A)} + e^{N^\alpha} \norm{A - B}\\
&\quad + C(B)e^{-C_\rho N^\eta}.
\end{align*}
Moreover, the constants $C(A), C(B)$ in the above depend on uniform measurements of $A$ and $B.$ In particular, the $L^2$-norms and constants from the Lojaciewicz inequality. Since $A$ and $B$ are close in norm, these constants are close, and we can, in fact, control them all by $2C(A).$ 

At this point, we set $N = \left(-\ln\norm{A - B}\right)^\beta,$ for some $\beta$ to be defined later. This yields
\begin{align*}
|L(A) - L(B)| &\leq C(A)(e^{-C_\rho (-\ln\norm{A - B})^{\beta\eta}}
 + e^{-2(-\ln\norm{A - B})^{\beta (\alpha - 1)}b(A)})\\
  &\quad + e^{(-\ln\norm{A - B})^{\beta\alpha}} \norm{A - B}.
\end{align*}

Now, if $\beta < 1/\alpha,$ then the $\norm{A - B}$ term above is bounded by $\norm{A - B}^{1/2},$ and is thus irrelevant, so we will take $\beta < 1/\alpha.$ We have
\begin{align*}
|L(A) - L(B)| &\leq C(A)e^{-C_\rho (-\ln\norm{A - B})^{\beta\eta}}\\
&\quad + C(A)e^{-2(-\ln\norm{A - B})^{\beta (\alpha - 1)}b(A)}.
\end{align*}

Finally, setting $\alpha = \eta + 1,$ we obtain our desired result with $\gamma = \beta \eta.$
\end{proof}

Weak-H\"older continuity in the frequency, at Diophantine frequencies, requires an additional estimate from our proof that the Lyapunov exponent is continuous at the frequencies with rationally independent components. We recall that below without proof.

\begin{mylemma}[\cite{PowellContinuity} Lemma 7.1]\label{Lem:FinalLemma}
Let $A(x)$ be an analytic quasi-periodic $M(2,\C)$-cocycle on $\T^d$ with an analytic extension to the strip $|\Im(z_j)| < \rho.$  Assume $\omega \in \T^{d}$ such that
$$\norm{k\cdot\omega} > \delta \quad \text{for } k\in \Z^{d_2}, 0 < |k| \leq K.$$
Moreover, suppose $\rho > K^{-c}.$
Then
$$|L_N - L| < K^{-c}$$
for all $N > K^2/\delta.$ Here $c = c(d) < 1$ is the same constant from before. 
\end{mylemma} 

Now we may obtain weak-H\"older continuity in $\omega.$

\begin{mythm}\label{Thm:FreqCont}
Let $A(x)$ be an analytic quasi-periodic $M(2,\C)$-cocycle on $\T^d$ with a plurisubharmonic extension to the strip $|\Im(z_j)| < \rho.$ Suppose that $\omega \in \T^d$ is such that $\norm{k\cdot \omega} > \tau |k|^{-\sigma} > 0$ for all $k\in\Z^d$ with $0 < |k|.$ Then there is $\delta > 0$ and $\gamma = \gamma(d,\sigma) \leq 1,$ such that, for every $\norm{\omega - \omega'} < \delta,$
\begin{equation}
|L(A,\omega) - L(A,\omega')| < C(A) \exp\set{-c(A)\left(-\ln\norm{\omega - \omega'}\right)^{\gamma}}.
\end{equation}
\end{mythm} 

\begin{Remark}
We note that $\delta$ arises in the form $\norm{\omega - \omega'} < e^{-cN_0^\beta} \sim \delta,$ and thus $\delta$ depends only on $N_0,$ which needs to be sufficiently large to apply Lemma \ref{Lemma:BaseStep}.
\end{Remark}

\begin{proof}
Since, this result involves a fixed cocycle, we will omit it in our notation. We have, for every $N > N_0,$ where $N_0$ is an absolute constant, and every $\alpha > 1,$
\begin{align*}
|L(\omega) - L(\omega')| &\leq |L(\omega') + L_N(\omega') - 2L_{2N}(\omega')|\\
&\quad + |L_N(\omega') + L_N(\omega)| + 2|L_{2N}(\omega') - L_{2N}(\omega)|\\
&\quad + |L(\omega) + L_N(\omega) - 2L_{2N}(\omega)|.
\end{align*}
Now, since $\omega$ satisfies the Diophantine condition $\norm{k \cdot \omega} > \tau|k|^{-\sigma},$ we may appeal to Theorem \ref{Thm:QualIteration} to control the last term. Moreover, the middle two terms may be controlled using Theorem \ref{Thm:FiniteCont}, and we obtain
\begin{align*}
|L(\omega) - L(\omega')| &\leq |L(\omega') + L_N(\omega') - 2L_{2N}(\omega')|\\
&\quad + C(A)e^{-2N^{\alpha - 1}b(A)} + e^{N^\alpha} N\norm{\omega - \omega'}\\
&\quad + C(A)e^{-2N^{\alpha - 1}b(A)} + e^{N^\alpha} N \norm{\omega - \omega'}\\
&\quad + C(A)e^{-C_\rho N^\eta}.
\end{align*}
Since $\omega'$ need not be Diophantine, we are unable to appeal to Theorem \ref{Thm:QualIteration} directly, however, we may use Corollary \ref{Cor:BaseStep} to obtain
$$|L_{N'}(\omega') + L_N(\omega') - 2L_{2N}(\omega')| < e^{-c'N^\eta}$$ for all $N' \leq  \exp(\exp(C_\rho N^\eta)).$

Now, if we assume $\norm{\omega - \omega'} < e^{-N^{1/\beta}},$ then we have 
$$\norm{k\cdot\omega'} > |k|^{-\sigma}$$
for all $0 < |k| \leq e^{N^{1/\beta}/(1 + \sigma)},$ and thus, for $N' = \exp(\exp(C_\rho N^\eta)),$ we may appeal to Lemma \ref{Lem:FinalLemma} to obtain
$$|L_{N'}(\omega') - L(\omega')| < e^{-cN^{1/\beta}}.$$
It follows that
$$|L(\omega') + L_N(\omega') - 2L_{2N}(\omega')| < e^{-cN^{1/\beta}} + e^{-c'N^\eta}.$$
Now we note that the condition on $\norm{\omega - \omega'}$ is equivalent to
$$N < (-\ln\norm{\omega - \omega'})^{\beta},$$
which is analogous to setting $N = (-\ln\norm{\omega - \omega'})^{\beta}$ for some $\beta > 0.$
Taking everything together, we are in precisely the setting from the proof of Theorem \ref{Thm:CocycleCont}, and we may proceed as we did there to conclude.
\end{proof}

Theorem \ref{Thm:MainThm} now follows by combining Theorems \ref{Thm:CocycleCont} and \ref{Thm:FreqCont}.

\section*{Acknowledgement}
We would like to thank W. Liu for useful comments on an earlier version of this work
and S. Jitomirskaya for fruitful discussions. This research was partially supported by NSF DMS-2052572, DMS-2052899, DMS-2155211, and Simons 681675.

\bibliographystyle{abbrv} 
\bibliography{WeakHolder}

\end{document}